\documentclass[a4paper,twocolumn,11pt,accepted=2021-11-09]{quantumarticle}

\pdfoutput=1

\usepackage{graphicx}
\usepackage{hyperref}
\usepackage{float}
\usepackage{amsmath,amssymb,amsthm,amsfonts,amstext}
\usepackage{url}
\usepackage{color}
\usepackage{braket}
\usepackage{titlesec}
\usepackage{soul}
\usepackage[numbers, sort]{natbib}

\usepackage{enumitem}
\usepackage{bbm}
\usepackage{changes}
\usepackage{lipsum}

\def\be{\begin{equation}}
\def\ee{\end{equation}}
\def\bea{\begin{eqnarray}}
\def\eea{\end{eqnarray}}
\def\bma{\begin{mathletters}}
\def\ema{\end{mathletters}}

\def\P{{\cal P}}
\def\A{{\cal A}}

\def\q0{\underline{0}}

\def\H{{\cal H}}
\def\Z{\mathbb{Z}}

\def\P{{\cal P}}

\def\C{{\mathbb C}}
\def\id{{\mathbb I}}

\def\H{{\cal H}}
\def\B{{\cal B}}

\def\N{\mathbb{N}}

\def\tr{\mbox{tr}}
\def\hier{\mathbb{H}}
\def\Ind{{\cal I}}
\def\one{\leavevmode\hbox{\small1\normalsize\kern-.33em1}}

\def\bra#1{\langle#1|} \def\ket#1{|#1\rangle}

\def\proj#1{\ket{#1}\!\bra{#1}}

\newtheorem{theo}{Theorem}

\newtheorem{defin}[theo]{Definition}

\newtheorem{prop}[theo]{Proposition}

\def\id{{\mathbb I}}

\newcommand{{\lev}}{L}

\begin{document}

\title{Entanglement marginal problems}

\author{Miguel Navascu\'es$^1$, Flavio Baccari$^2$ and Antonio Ac\'in$^{3,4}$}
\affiliation{$^1$Institute for Quantum Optics and Quantum Information (IQOQI) Vienna, Austrian Academy of Sciences\\
$^2$Max-Planck-Institut f\"ur Quantenoptik, Hans-Kopfermann-Straße 1, 85748 Garching, Germany\\
$^3$ICFO-Institut de Ciencies Fotoniques, The Barcelona Institute of Science and Technology, 08860 Castelldefels (Barcelona), Spain\\
$^4$ICREA-Institucio Catalana de Recerca i Estudis Avan\c cats, Lluis Companys 23, 08010 Barcelona, Spain}

\begin{abstract}
We consider the entanglement marginal problem, which consists of deciding whether a number of reduced density matrices are compatible with an overall separable quantum state. To tackle this problem, we propose hierarchies of semidefinite programming relaxations of the set of quantum state marginals admitting a fully separable extension. We connect the completeness of each hierarchy to the resolution of an analog classical marginal problem and thus identify relevant experimental situations where the hierarchies are complete. For finitely many parties on a star configuration or a chain, we find that we can achieve an arbitrarily good approximation to the set of nearest-neighbour marginals of separable states with a time (space) complexity polynomial (linear) on the system size. Our results even extend to infinite systems, such as translation-invariant systems in 1D, as well as higher spatial dimensions with extra symmetries.
\end{abstract}

\maketitle

\section{Introduction}
Recent quantum experiments with cold atoms and superconducting qubits have demonstrated different degrees of control over systems of size ranging from about $50$ \cite{supremacy} to several hundred qubits \cite{schmied}. Determining whether entanglement is present in such architectures is a challenging task. To begin with, a full tomographic reconstruction of the underlying quantum state is out of the question, since it requires the estimation of a number of parameters that grows exponentially with the system size. A more realistic goal consists in estimating a number of parameters that scales polynomially with the system size, for instance the reduced density matrices of neighboring subsystems. 

In that case, however, most general methods for entanglement detection fail, since they are tailor-made for scenarios where the whole state is available and therefore require an optimization over an, again, exponential number of parameters. For instance, the application of the Doherty-Parrilo-Spedalieri (DPS) hierarchy of relaxations \cite{DPS1} to characterize the set of separable states, or some entanglement detection methods based on incomplete information~\cite{GRW} would require solving an optimization problem where one of the variables corresponds to the full density matrix of the system. The memory cost of such approaches becomes prohibitive already beyond the case of order of ten qubits. On the other hand, there exist entanglement detection methods based on two-point correlation functions~\cite{spinSq1,spinSq2}, but no general construction is known for them and they are not guaranteed to detect all entangled states. Leaving aside entanglement detection through indirect methods such as Bell inequality violations \cite{tura2017energy, Aloy_2019}, to our knowledge there is no systematic and scalable method to date to detect entanglement from incomplete information.

In this work, we address this question and consider the \emph{entanglement marginal problem}: deciding if an ensemble of density matrices could be the marginal states of a separable global state $\rho$. We present a framework to construct hierarchies of tests that are satisfied whenever the ensemble of reduced states is compatible with a global separable state. The tests amount to solving a semidefinite program (SDP) \cite{sdp}, a class of convex optimization problems that can be solved efficiently. If the ensemble fails one of the tests, then our method will output a linear entanglement witness in terms of the states in the ensemble that certifies the entanglement of $\rho$. We show that the convergence of the hierarchy is strongly connected to the classical marginal problem. For those cases in which the hierarchy is complete, we prove that any ensemble passing the ${\lev}^{th}$ test of the hierarchy is $O(1/{\lev}^2)$-close in trace norm to the set of ensembles admitting a separable extension. 

To our surprise, we find that, in many cases of physical interest, the corresponding hierarchies have a time complexity polynomial on the system size, and a memory complexity linear on the system size. This allows us to certify entanglement in systems of hundreds of sites. Our hierarchies also apply to characterize entanglement in infinite systems, as long as those are subject to some natural symmetries. More precisely, we propose two SDP hierarchies that, respectively, characterize entanglement in 1D translation-invariant (TI) systems and in 2D TI systems in the square lattice that are also symmetric under horizontal (or vertical) reflections. 
We test the practical performance of our approach by working out the critical temperature beyond which no entanglement is present in the system for a class of solvable models. We also compute the so-called \emph{separable energy per site}, the minimum energy per site achievable by separable states, for some $k$-local Hamiltonians.

This article is organized as follows: in Section \ref{notation} we introduce the notation and concepts to be used along the text. In Section \ref{entanglementmarg} we define the entanglement marginal problem, the subject of this work, and we relate it to the classical marginal problem. In section \ref{hierSec}, we present a hierarchy of SDPs to attack the entanglement marginal problem. We fully characterize, in Proposition \ref{conv_prop}, the set of state ensembles to which this hierarchy converges to. As we argue, this set corresponds to the set of marginals of an overall separable state iff the classical marginal problem is trivial in the considered marginal scenario. Otherwise, it is just a strict relaxation. We then apply our construction to solve the entanglement marginal problem in scenarios with finitely many parties (section \ref{finiteSec}) or symmetric systems with infinitely many parties (section \ref{infiniteSec}). Finally, in section \ref{conclusion} we present our conclusions.

\section{Preliminaries}
\label{notation}
Let ${\A}= (\alpha,\beta,...)$ be a finite or an infinite countable alphabet. In the following, each of the letters $\alpha\in\A$ will represent either a quantum system with Hilbert space dimension $d_\alpha$ or a classical random variable with finite or infinite cardinality. For finite systems of $n$ particles, one can take ${\A}$ to be the set $\{1,\ldots,n\}$, but our considerations also apply to infinite systems, which motivates the previous notation. Let $I\subset \A$ be a subset of all such systems. For any vector or sequence $\phi\equiv (\phi_\alpha)_{\alpha\in\A}$, we call $\phi_I$ the vector or sequence with coefficients $(\phi_\alpha: \alpha\in I)$. Similarly, for any quantum state $\rho:\H\to\H$, with $\H=\bigotimes_{\alpha\in\A} \H_\alpha$ being a composite Hilbert space, the operator $\rho_I$ denotes the partial trace $\tr_{\A\setminus I}(\rho)$. This notation extends to probability distributions: $p_I(\phi_I)$ denotes the marginal distribution of $p(\phi)$ for variables $(\phi)_{\alpha\in I}$. The last two sentences do not make much sense mathematically when $|\A|=\infty$; in that case, we refer the reader to Appendix \ref{infinite_systems_app} for a definition of quantum states, probability distributions and their respective marginalizations.

Let $\Ind\subset P(\A)$, where $P(\A)$ denotes the power set of $\A$, that is, the set consisting of all its subsets. By conducting several tomographic experiments on all the systems in $I\in \Ind$, we can estimate $\rho_I$, the density matrix describing the sites at $I$. A necessary condition for the existence of an overall quantum state $\rho$ for all systems in $\A$ compatible with the ensemble $\{\rho_I\}_{I\in\Ind}$ is that of \emph{local compatibility}, i.e.,

\be
\tr_{I\setminus J}(\rho_{I})=\tr_{J\setminus I}(\rho_{J}),
\ee
\noindent for $I,J\in \Ind$.

Local compatibility can also be defined for ensembles of probability distributions. If $\{p_{I}(\phi_I)d\phi_I:I\in\Ind\}$ are the marginals of a global measure $p(\phi)$, defined over all sites in $\A$, then $\{p_{I}(\phi_I)d\phi_I:I\in\Ind\}$ must satisfy the local compatibility constraints

\be\label{eq:localcompat}
\int d\phi_{I\setminus J} p_{I}(\phi_I)=\int d\phi_{J\setminus I}p_{J}(\phi_J), \forall I,J\in\Ind\,
\ee
\noindent for $I,J\in \Ind$.

In section \ref{infiniteSec} we study configurations where the physical systems lie on the sites of an infinite chain or an infinite hypercubic lattice and the overall state describing $\A$ satisfies translation invariance. In those scenarios, local compatibility admits a very simple form. Consider, for instance, an infinite translation-invariant 1D system, with $\A=\Z$ and $d_\alpha=d$ for all $\alpha\in\A$. Let $\omega_{\{1,...,k\}}$ denote the density matrix (probability distribution) describing the quantum systems (classical random variables) at sites $1,...,k$. If $\omega_{\{1,...,k\}}$ is the marginal of a translation-invariant 1D system $\omega$, then $\omega_{\{1,...,k\}}$ must also be subject to the local compatibility relations

\be
\omega_{\{1,...,k-1\}}=\omega_{\{2,...,k\}}.
\ee
\noindent This condition is known as local translation invariance (LTI) \cite{TI}.

In higher spatial dimensions, translation invariance enforces slightly more complicated local constraints. Think of a 2D infinite square lattice, i.e., $\A=\Z^2$. A $k\times l$ plaquette is a set of sites of the form $\hat{I}+z$, for $\hat{I}=\{(x,y): x=1,...,k; y=1,...,l\}$, $z\in\Z^2$. If the quantum or classical description $\omega_{\hat{I}}$ of the sites in $\hat{I}$ is the marginal of a TI quantum or classical system, then local compatibility implies that 

\begin{align}
&\omega_{\hat{I}\setminus\{(1,y):y\}}=\omega_{\hat{I}\setminus\{(k,y):y\}},\nonumber\\
&\omega_{\hat{I}\setminus\{(x,1):x\}}=\omega_{\hat{I}\setminus\{(x,l):x\}},
\end{align}
\noindent see Figure \ref{LTI_pic}. These conditions are also dubbed LTI~\cite{TI}.

\begin{figure}
  \centering
  \includegraphics[width=8cm]{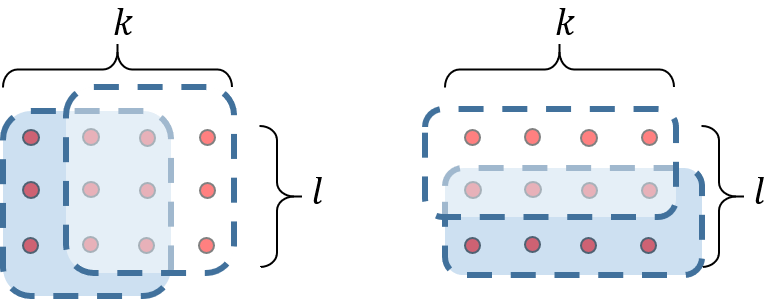}
  \caption{\textbf{A $k\times l$ plaquette.} Local translation invariance implies that the marginals (quantum or classical) corresponding to the sites enclosed by the two dashed lines in each figure are identical.}
  \label{LTI_pic}
\end{figure}

\section{The entanglement marginal problem}% and its  connection to the classical one}
\label{entanglementmarg}

A multi-partite quantum state $\sigma$, with subsystems labeled by the elements of ${\cal A}$, is said to be separable if there exists a multi-variate measure $p(\phi)d\phi$, with $\phi =(\phi_\alpha)_{\alpha\in\A}$, such that

\be
\sigma = \int p(\phi)d\phi \bigotimes_{\alpha\in\A}\proj{\phi_\alpha}
\label{separable}
\ee
\noindent (if $|\A|=\infty$, we need to marginalize $p$ over finite subsets of $\A$, see Appendix \ref{infinite_systems_app}). In the following, when we refer to a measure $p(\phi)$ generating a separable state $\sigma$, we implicitly mean that the state is defined by Eq.~\eqref{separable}.

In this work, we wish to solve what from now on we call the \emph{entanglement marginal problem}, stated as follows. 

\begin{defin}
\noindent\textbf{Entanglement marginal problem.}

Input: $\A,\Ind\subset P(\A)$, an ensemble $\{\rho_I\}_{I\in\Ind}$ of locally compatible states. 

Goal: find whether there exists a separable state $\sigma$ such that $\tr_{\A\setminus I}(\sigma) =\rho_I$ for all $I\in\Ind$.  
\end{defin}

As we will see, the entanglement marginal problem is strongly connected to the \emph{classical marginal problem}, a very old question that appears in numerous contexts in physics, such as statistical physics \cite{TI, undecidable} or quantum nonlocality \cite{oneD}. 

\begin{defin}
\noindent\textbf{Classical marginal problem.}

Input: $\A,\Ind\subset P(\A)$, an ensemble $\{p_I(\phi_I)\}_{I\in\Ind}$ of locally compatible probability distributions. 

Goal: find whether there exists a global measure $\tilde{p}$ such that $\tilde{p}_I(\phi_I) =p_I(\phi_I)$ for all $I\in\Ind$.
\end{defin}

In general marginal scenarios $\A,\Ind\in P(\A)$, local compatibility of classical distributions does not always imply the existence of a global measure. Consider, for instance, $\A=\{1,2,3\}$, and the distributions of random variables $a_1,a_2,a_3\in\{0,1\}$ given by $P_I(0,1)=P_I(1,0)=\frac{1}{2}$, for $I = \{1,2\}, \{1,3\},\{2,3\}$. These distributions satisfy local compatibility, and yet it is easy to see that no global distribution $P(a_1,a_2,a_3)$ admits them as marginals. The classical marginal problem is thus non-trivial, and, in fact, it can be shown to be NP-hard \cite{pitowsky1989quantum}.

In order to understand the connection between the entanglement marginal problem and its classical counterpart, let us rewrite the reduced states $\rho_I$ as
\be
\rho_{I}\equiv\int p_{I}(\phi_I)d\phi_I\bigotimes_{\alpha\in I}\proj{\phi_\alpha}.
\label{separablemarg}
\ee
for some given probability measures $p_{I}(\phi_I)$. Notice that one can assume such a decomposition always to exist, since otherwise some of the reduced states would be entangled, making the solution to the entanglement marginal problem trivial. Solving the entanglement marginal problem is equivalent to determining whether it is possible to find a collection of distributions $\lbrace p_{I}(\phi_I) \rbrace_{I\in\Ind}$ that admits a global measure and such that eq. (\ref{separablemarg}) holds. Indeed, if such a measure $q(\phi)$ exists, it could be used to generate a separable state $\sigma$ whose reduced states on each set $I\in\Ind$ coincides with $\rho_{I}$, simply because the corresponding marginal distribution of $q(\phi)$ coincides with $p_{I}(\phi_I)$.

The connection between the classical and the entanglement marginal problem does not make either of them less challenging. First of all, because the choice of probability measures $p_{I}(\phi_I)$ in \eqref{separablemarg} is generally not unique and there is no guarantee for those measures even to be locally compatible, in the sense of \eqref{eq:localcompat}. Moreover, even if we could enforce their local compatibility, this is generally not enough to guarantee the existence of a global measure, and hence of an overall separable extension for the states. We come back to this point below. However, and despite the complexity of the problems involved, relating the entanglement marginal problem and its classical counterpart offers a useful interpretation to the hierarchies that we introduce in the next Section, which will prove crucial to analyze their convergence properties. 

%In particular, we propose hierarchies of necessary conditions for the existence of a separable extension that can also be  seen as a relaxation of the above problem. As we will see, the properties of the corresponding classical marginal problem will prove crucial for the convergence of the presented hierarchies.}

%Sometimes, we will require $\sigma$ to satisfy some symmetries, such as translation invariance or invariance under reflections of the axes. In that predicament, the considered ensemble $\{\rho_I\}_I$ will also be presumed to satisfy LTI.

\section{Semidefinite programming relaxations of the entanglement marginal problem}
\label{hierSec}

To solve the entanglement marginal problem, we follow the ideas of Doherty Parrilo and Spedalieri (DPS) ~\cite{DPS1,DPS2,DPS3}, who provided an algorithm to solve the entanglement problem in the case where the whole density matrix $\rho$ is specified, i.e., when ${\cal I}=\{\A\}$. 
Suppose then that there exists a probability distribution $p(\phi)$ over pure states $\phi$ such that the separable state defined by Eq.~(\ref{separable}) admits the states $\{\rho_I\}$ as marginals. Next, define the states

\be
\rho^{({\lev})}_{I}\equiv\int p_{I}(\phi_I)d\phi_I\bigotimes_{\alpha\in I}\proj{\phi_\alpha}^{\otimes {\lev}}.
\label{truqui}
\ee

\noindent State $\rho^{({\lev})}_{I}$ thus acts on a composite Hilbert space whose factors are the ${\lev}^{th}$ tensor powers of the Hilbert space of each system $\alpha\in I$. Notice that each of them can be interpreted as a symmetric extension of $\rho_{I}$ derived from the decomposition \eqref{separablemarg}. Denoting by $I^k$ the first $k$ copies of each system in $I\in\Ind$, it is easy to see that the states $\{\rho^{({\lev})}_{I}\}_{I\in\Ind}$ satisfy the constraints:

\begin{enumerate}[label=(\roman*)]
\item
$\tr_{I^{{\lev}-1}}(\rho^{({\lev})}_{I})=\rho_{I}$.

\item
$\rho^{({\lev})}_{I}$ is Positive under Partial Transposition (PPT) \cite{peres} across all bipartitions of its $|I|{\lev}$ systems.
\item
$\rho^{({\lev})}_{I}\in B(\bigotimes_{\alpha\in I}\H_\text{sym}({\lev},d_\alpha))$, where $\H_\text{sym}({\lev},d)$ denotes the symmetric space of ${\lev}$ $d$-dimensional particles.
\item
%$\tr_{\alpha^N: \alpha\in I\setminus J}(\rho^{(N)}_{I})=\tr_{\alpha^N: \alpha\in J\setminus I}(\rho^{(N)}_{J})$.
$\tr_{(I\setminus J)^{\lev}}(\rho^{({\lev})}_{I})=\tr_{(J\setminus I)^{\lev}}(\rho^{({\lev})}_{J})$ (after appropriately reordering the systems of one of the sides).
%\item
%For all $({\cal P}, I, c_{{\cal P}})\in\rest$, $\sum_{\pi\in{\cal P}}c_\pi \rho^{(N)}_{\pi(I)}=0$.
\end{enumerate}

For each ${\lev}$, we relax the property that $\{\rho_{I}\}_{I}$ arise from an overall separable state by demanding that there exist positive semidefinite matrices $\{\rho_I^{({\lev})}\}_I$ satisfying the four conditions above (note that the last condition is missing in the DPS construction \cite{DPS3}). For finite $\A$, the existence of the matrices $\{\rho_I^{({\lev})}\}_I$ can, in turn, be cast in terms of a semidefinite program (SDP)~\cite{sdp}. As we increase ${\lev}$, these conditions lead to the announced SDP \emph{hierarchy} $\hier (\Ind|\{\rho_I\}_{I\in\Ind})$ ($\hier$, for short) of tests for entanglement detection. 

\begin{defin}
\noindent\textbf{SDP hierarchy $\hier$}

Given a marginal scenario $\A,\Ind\in P(\A)$ and $\lev\in\mathbb{N}$, $\hier^\lev$ is the set of state ensembles $\{\rho_I\}_{I\in\Ind}$ for which there exist $\{\rho^{({\lev})}_{I}\}_{I\in\Ind}$ such that, for all $I,J\in\Ind$,

\begin{align}
&\rho^{({\lev})}_{I}\in B\left(\bigotimes_{\alpha\in I}\H_\text{sym}({\lev},d_\alpha)\right),\nonumber\\
&(\rho^{({\lev})}_{I})^{T_S} \, \succeq \, 0, \forall S\subset I^L, \nonumber\\
&\tr_{I^{{\lev}-1}}(\rho^{({\lev})}_{I})=\rho_{I}, \nonumber\\ 
&\tr_{(I\setminus J)^{\lev}}(\rho^{({\lev})}_{I})=\tr_{(J\setminus I)^{\lev}}(\rho^{({\lev})}_{J}).
\end{align}
\noindent If $\{\rho_I\}_{I\in\Ind}\in\hier^{\lev}$, we say that the ensemble $\{\rho_I\}_{I\in\Ind}$ passes the $L^{th}$ test of hierarchy $\hier$. Otherwise, $\{\rho_I\}_{I\in\Ind}$ is said to fail the $L^{th}$ test.
\end{defin}
Note that the ensemble variables $\{\rho^{({\lev})}_I\}_{I\in\Ind}$ and their partial transposes can be embedded, through direct sum, into an overall positive semidefinite variable, whose diagonal blocks are subject to the linear/affine constraints above. That way, for each $\lev$, one arrives at the canonical form of a semidefinite program \cite{sdp}.

Each level of the hierarchy $\hier$ involves matrices acting on $\bigotimes_{\alpha\in I}\H_\text{sym}({\lev},d_\alpha)$, see condition (iii). Since the dimension of $\H_\text{sym}({\lev},d)$ is
\begin{equation}
\label{dsym}
d_\text{sym}=\left(\begin{array}{c} {\lev}+d-1\\d-1\end{array}\right)\sim O({\lev}^{d-1})\,,
\end{equation}
the size of the matrices $\rho^{({\lev})}_{I}$ scales as $O({\lev}^{\sum_{\alpha\in I}(d_\alpha-1)})$. Clearly, if the ensemble $\{\rho_I\}_{I\in {\cal I}}$ fails the ${\lev}^\text{th}$ SDP test of hierarchy $\hier$, then it cannot admit a separable extension. In such a predicament, the dual of the SDP \cite{sdp} will return a certificate of infeasibility, i.e., an entanglement witness in the form of a linear inequality $\sum_I\tr(W_I\rho_I)\geq 0$ that is violated by the ensemble $\{\rho_I\}_I$.

In Appendix \ref{simplified}, we present also an alternative simplified hierarchy of SDPs $\bar{\hier}(\Ind)$ for those cases in which $I\in\Ind$ contains an index $\alpha_I$ that does not appear in any other set $J\in\Ind$, $J\not=I$. One can then provide a hierarchy where the system $\alpha_I$ is not ${\lev}$-times ``extended'' in the construction of the matrix variable $\rho^{({\lev})}_I$. Trivially, any ensemble $\{\rho_I\}_I$ passing the ${\lev}^\text{th}$ test of $\hier(\Ind)$ will also pass the ${\lev}^\text{th}$ test of $\bar{\hier}(\Ind)$. For some scenarios $\Ind$, though, both hierarchies provide a similar approximation to the set of separable marginals and implementing $\bar{\hier}^{\lev}$ requires considerably fewer computational resources.

The natural question to ask is whether the SDP hierarchies defined here are complete, i.e., under which conditions $\{\rho_{I}\}_{I}\in \hier^{\lev}$ for all ${\lev}$ implies that $\{\rho_{I}\}_{I}$ admits a separable extension. To answer it, we exploit the connection to the classical marginal problem thanks to the following proposition, proven in Appendix~\ref{proofProp} and based on the linear maps introduced in~\cite{NOP}.

\begin{prop}
\label{conv_prop}
Let $\{\rho_I\}_I\in {\hier}^{\lev}$ or $\{\rho_I\}_I\in \bar{\hier}^{\lev}$. Then, there exists an ensemble of fully separable states $\{\tilde{\rho}_I\}_I$, such that, for all $I\in\Ind$,

\begin{align}
\|\rho_I-\tilde{\rho}_I\|_1\leq 2\left\{1-\prod_{\alpha\in I,\alpha\not=\alpha_I}(1-\epsilon(\lev,d_\alpha))\right\},
\label{poly_approx}
\end{align}
\noindent where

\begin{align}
\epsilon(\lev,d)=&\frac{d}{2(d-1)}\times\nonumber\\
&\min\{1-x:P^{d-2, \lev \pmod   2}_{\lfloor\lev/2 +1\rfloor} (x)=0\},
\label{epsilon_def}
\end{align}
\noindent and $\{P^{\beta,\gamma}_m(x)\}_{\beta,\gamma,m}$ denote the Jacobi polynomials \cite{stegun}.

Moreover, the separable states $\{\tilde{\rho}_I\}_I$ are generated by an ensemble $\{\tilde{p}_I(\phi_I)\}_I$ of locally compatible distributions.
\end{prop}

By the properties of the zeros of the Jacobi polynomials \cite{stegun}, it follows that, for $\lev\gg d_\alpha$, the right-hand side of (\ref{poly_approx}) tends to $O\left(\frac{\sum_{\alpha\in I,\alpha\not=\alpha_I}d_\alpha^2}{\lev^2}\right)$. The above proposition thus shows that the introduced hierarchies are a converging series of relaxations of the set of reduced states $\{\rho_I\}_I$ admitting a decomposition \eqref{separablemarg} generated by locally compatible probability measures. 

However, since these compatibility conditions do not necessarily imply a solution to the classical marginal problem, the proposition is not a proof of convergence, as the states $\{\tilde{\rho}_I\}_I$ are not necessarily the marginals of an overall separable state $\sigma$ as in eq. (\ref{separable}). In fact, suppose that there exists an ensemble $\{\tilde{p}_I(\phi_I)\}_I$ of locally compatible distributions which are \emph{not} generated by an underlying global measure. Then, for $m$ sufficiently high, one can find a similar counterexample with an ensemble of probability distributions $\{\tilde{P}_I(a_I)\}$ of random variables $\{a_\alpha\}_{\alpha\in\A}$ taking values in $\{1,...,m\}$. In turn, such distributions can be used to build an ensemble $\rho^{\text{noext}}_I$ of classical states of the form $\rho^{\text{noext}}_I=\sum_{a_I}\tilde{P}(a_I)\proj{a_I}$, where $\{\ket{a_I}: a_I\in \{1,...,m\}^{|I|}\}$ is the computational basis of $(\C^{m})^{\otimes |I|}$. This ensemble is generated by locally compatible probability measures and yet there is no overall quantum state $\sigma$ (separable or not) admitting $\{\rho_I\}_I$ as reduced density matrices. For these cases, the hierarchy $\hier$ is incomplete.

Conversely, Proposition $1$ guarantees convergence in all those classical marginal scenarios $\Ind$ where global compatibility follows from local compatibility. In Section \ref{finiteSec} and \ref{infiniteSec} we present several such scenarios for systems of finite and infinite size, and analyse the corresponding entanglement marginal problems. These examples illustrate the applicability of hierarchies $\hier,\bar{\hier}$ to detect entanglement in relevant experimental situations.

\section{Finite systems}
\label{finiteSec}

In this section we present marginal scenarios for finite systems, $\Ind\subset P(\A)$, with $|\A|<\infty$, where the hierarchies $\hier,\bar{\hier}$ completely solve the entanglement marginal problem. 

First, we need to introduce some graph notation. A graph ${\cal G}= (V,E)$ is defined by a collection of indices $V$, called vertices, together with a set $E$ of pairs of vertices $\{v_i,v_j\}$, called edges. Two vertices $v_i, v_j$ are said to be connected if there exists an edge $e\in E$ such that $v_i,v_j\in e$. A \emph{clique} is a set of pair-wise connected vertices, and a \emph{maximal clique} $C$ is a clique such that, for any $v\in V, v\not\in C$, $C\cup \{v\}$ is not a clique. A \emph{graph cycle} is a sequence of vertices $(v_i)_{i=1}^m$ in $V$ such that $v_1=v_m$ and $v_i,v_{i+1}$ are connected for $i=1,...,m-1$. A \emph{chord} in a graph cycle is an edge that does not form part of the cycle but nonetheless connects two vertices of it, and a \emph{chordal graph} is a graph such that any cycle of length $4$ or higher has a chord. Chordal graphs find several applications in semidefinite programming and combinatorial optimization because of their many appealing properties.  It is well known that it is very efficient to find both a perfect elimination ordering for the vertices of a chordal graph and an arrangement of its maximal cliques so to fulfil the so-called \emph{running intersection property}.  While these properties are related (see for instance \cite{vandenberghe2015chordal}),  the second has a clearer connection to the entanglement marginal problem,  as we will show in the following.

Given $\A,\Ind$, define a graph ${\cal G}$ with vertices given by the letters of $\A$, and join with an edge every pair of vertices $\alpha,\beta$ such that $\alpha,\beta\in I$ for some $I\in {\cal I}$. In the following, we will refer to ${\cal G}$ as the \emph{dependency graph} of $\Ind$. The following proposition provides a sufficient condition for a marginal scenario $\A,\Ind\in P(\A)$ to have a trivial classical marginal problem.

\begin{prop}
\label{trivial_classical}
Let $\A,\Ind\in P(\A)$ be a marginal scenario, and let ${\cal G}$ be its corresponding dependency graph. If ${\cal G}$ is chordal and the elements of $\Ind$ are its maximal cliques, then any ensemble of locally compatible probability distributions $\{\tilde{p}_I\}_{I\in\Ind}$ admits a global measure.
\end{prop}

\noindent This proposition is a well-known fact in combinatorial optimization \cite{laserre} and relates to the running intersection property mentioned above. For completeness, we include a proof in Appendix \ref{runningApp}. Combining Propositions \ref{conv_prop} and \ref{trivial_classical}, we arrive at our first main result.

\begin{theo}
\label{charac_chordal}
Let $\A,\Ind\in P(\A)$ be a marginal scenario such that the dependency graph of $\Ind$ is chordal and each $I\in\Ind$ is one of its maximal cliques. Let $\{\rho_I\}_{I\in\Ind}\in \bar{\hier}^\lev$. Then, there exists a separable state $\sigma$ such that 

\be
\|\sigma_I-\rho_I\|_1\leq O\left(\frac{\sum_{\alpha\in I, \alpha\not=\alpha_I}d_\alpha^2}{\lev^2}\right),
\ee
\noindent for all $I\in\Ind$.
\end{theo}

As we argue in the next section, there exist relevant experimental scenarios with a chordal dependency graph. For those, the above theorem guarantees that $\bar{\hier}$ provides a complete entanglement characterization. 

Moreover, Theorem \ref{charac_chordal} can be useful to characterize the set of separable ensembles even in marginal scenarios where the dependency graph is not chordal. Indeed, suppose that we wish to decide whether the ensemble of density matrices $\{\rho_I\}_{I\in\Ind}$ are the marginals of an overall separable state, but the dependency graph ${\cal G}$ of $\Ind$ is not chordal. Let $\bar{{\cal G}}$ be a chordal graph with vertices $\A$ such that ${\cal G}$ is a subgraph of $\bar{{\cal G}}$, and call $\bar{\Ind}$ the set of maximal cliques of $\bar{{\cal G}}$. Since ${\cal G}$ is a subgraph of $\bar{{\cal G}}$, it follows that there exists a function $c:\Ind\to\bar{\Ind}$ such that $I\subset c(I)$, for all $I\in\Ind$. Thus, by Theorem \ref{charac_chordal}, the following  hierarchy of SDPs fully characterizes separability in $\Ind$:

\begin{align}
&\exists \{\rho_I\}_{\bar{I}\in \bar{\Ind}}\in \bar{\hier}^\lev(\bar{\Ind}), \nonumber\\
\mbox{such that }&\tr_{c(I)\setminus I}(\rho_{c(I)})=\rho_{I},\forall I\in\Ind.
\label{completion}
\end{align}

In graph-theoretic terminology, we would say that $\bar{{\cal G}}$ is a \emph{chordal completion} of ${\cal G}$, and, to minimize the time and space complexity of the hierarchy above, we need that the size of the maximal clique of $\bar{{\cal G}}$ be as small as possible. Finding the chordal completion of an arbitrary graph with minimum maximum clique size is known to be an NP-complete problem \cite{tree_width}. However, for fixed $k\in\N$, deciding whether a graph admits a chordal completion of maximum clique size $k$ or less is a problem that can be solved in linear time on $|\A|$ \cite{tree_width_2}. Since $k$ is fixed by our computer memory, it follows that, for any entanglement marginal problem, we can use the algorithm in \cite{tree_width_2} to assess if there exists a tractable SDP hierarchy of the form (\ref{completion}) to tackle it.

Graphs admitting a chordal completion with maximum clique size $k+1$ are called \emph{partial $k$-trees}, and they have been studied extensively. They include cactus graphs, pseudoforests, series-parallel graphs, outerplanar graphs, Halin graphs and Apollonian networks \cite{bounded_tree_width}. If the dependency graph of a given marginal scenario belongs to one of such families, then one can solve the corresponding entanglement marginal problem efficiently. 
As we will analyse in detail next, the above families of graphs comprise some experimentally relevant scenarios for entanglement detection in many-body quantum systems.

\subsection{Particles in a star configuration}
Consider first the scenario where we have access to all the two-body reduced states $\lbrace \rho_{1j} \rbrace_{j = 2}^n$ of particle $1$ and each of the remaining particles $\{2,...,n\}$. Determining whether the overall state of the system is entangled amounts to solving the entanglement marginal problem for $\Ind=\{I_j:j=2,...,n\}$, with $I_j\equiv \{1,j\}$. The corresponding dependency graph is shown in Figure \ref{star_pic}.

\begin{figure}
  \centering
  \includegraphics[width=5cm]{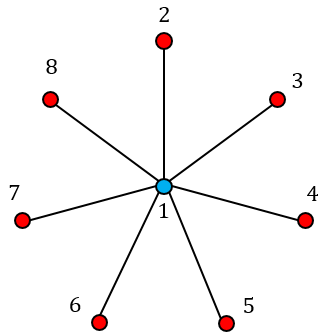}
  \caption{\textbf{The star configuration.} The dependency graph consists of the edges $I_j\equiv \{1,j\}$, for $j=\{2,...,n\}$.}
  \label{star_pic}
\end{figure}

By simple inspection, we notice that the graph is chordal. Moreover, for $j\not=1$, system $j$ only appears in $I_j$. Hence, by Proposition \ref{conv_prop}, the hierarchy $\bar{\hier}$, with $\alpha_j=j$ for $j=2,...,n$, is convergent. Writing it explicitly, the SDP corresponding to the ${\lev}$-th level of $\bar{\hier}$ reads as follows
\begin{eqnarray}\label{eq:SDPstar}
 & \exists \, \, & \lbrace \rho_{a_1a_2\ldots a_{\lev}b|j} \in B(\H_\text{sym}({\lev},d_1)\otimes\H_{d_j})\rbrace_{j = 2}^n \, , \nonumber \\
& \quad\text{s.t.} &  \forall \, \, j \, \in \, \lbrace 2,\ldots, n \rbrace  : \nonumber \\
 & \text{(i)} & \rho_{a_1a_2\ldots a_{\lev}b|j}\succeq 0\,, \nonumber \\
 & \text{(ii)} & \tr_{(a_2\ldots a_{\lev})} (\rho_{a_1a_2\ldots a_{\lev}b|j}) = \rho_{1j}  \\
% & \text{(ii)}  & \rho_{\pi(a_1a_2\ldots a_N) b|j} = \rho_{a_1a_2\ldots a_Nb|j} \quad \forall \, \, \pi \in S_N \nonumber  \\
 & \text{(iii)} &  \rho_{a_1a_2\ldots a_{\lev} b|j}^{T_{a_1\ldots a_k}} \succeq 0 \quad k = 1,\ldots,{\lev} \nonumber \\
 &\text{and} & \forall \, \, k \, \in \lbrace 3,\ldots n \rbrace \, : \nonumber \\
 &\text{(iv)} & \tr_{b} (\rho_{a_1a_2\ldots a_{\lev}b|2}) = \tr_{b} (\rho_{a_1a_2\ldots a_{\lev}b|k}) \nonumber
\end{eqnarray}
Note that, in this case, implementing $\hier^{\lev}(\Ind)$ would require optimizing over positive semidefinite matrix variables $\rho^{({\lev})}_{I}$ of size $O({\lev}^{d_1+d_j-2})$, each of them acting on the space $\H_\text{sym}^{{\lev},d_1}\otimes\H_\text{sym}^{{\lev},d_j}$. Running $\bar{\hier}^{\lev}(\Ind)$, on the contrary, just requires extending system $1$, and so the corresponding positive semidefinite matrix variables are of size $O({\lev}^{d_1-1})$.

Let us then see how the hierarchy \eqref{eq:SDPstar} performs in practice. Suppose that all systems are qubits ($d_j=2$ for $j=1,...,n$), and $n=4$. Let $X,Y,Z$ denote the Pauli matrices, and consider the problem of computing the separable energy of the Hamiltonian $H=X_1X_2+Y_1Y_3+Z_1Z_4$, i.e., we wish to minimize the term 

\be
\tr(\rho_{I_2}X\otimes X)+\tr(\rho_{I_3}Y\otimes Y)+\tr(\rho_{I_4}Z\otimes Z)
\ee
\noindent over all ensembles $\{\rho_I:I\in\Ind\}$ admitting a separable extension.

An SDP optimization over $\{\rho_I:I\in\Ind\}\in\bar{\hier}^2$ returns the exact value $-\sqrt{3}\approx -1.7320$, achievable by the product state 
\be
\sigma=\frac{\id+\frac{X+Y+Z}{\sqrt{3}}}{2}\otimes\frac{\id-X}{2}\otimes\frac{\id-Y}{2}\otimes\frac{\id-Z}{2}.
\ee
\noindent This optimization, and all the following ones, was carried out using the MATLAB optimization package YALMIP \cite{yalmip}, in combination with the SDP solver Mosek \cite{mosek}.

It is interesting to compare the exact value $-\sqrt{3}$ with that obtained by minimizing $H$ instead over ensembles of separable states $\{\rho_I:I\in\Ind\}$ satisfying local compatibility. For two-qubit systems, the PPT condition is equivalent to separability \cite{Horodecki_1996}: this allows us to carry out such an optimization exactly via SDP. The result is $-3$. This example shows that the set of locally compatible separable states can be a very bad outer approximation to the set of ensembles admitting a separable extension. Moreover, as proven in \cite{marginal1,marginal2, oneD}, this is so even when one demands in addition the existence of an overall quantum state for the state ensemble. The need of constructions like our hierarchies to address the entanglement marginal problem is thus vindicated.

\subsection{Particles in a line}
\label{lineSec}

\begin{figure}
  \centering
  \includegraphics[width=8cm]{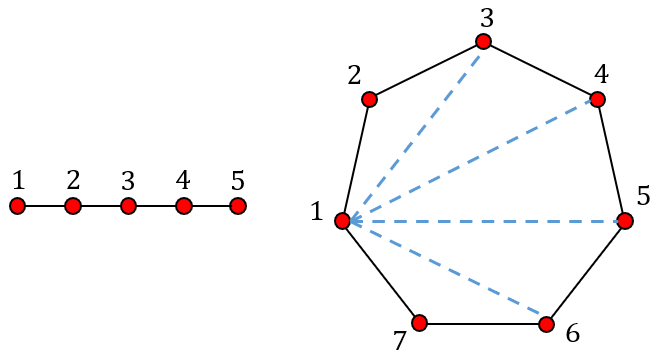}
  \caption{\textbf{The 1D case.} Left: the dependency graph of a 1D chain with open boundary conditions, given access to nearest neighbors. Right: the dependency graphs for the 1D case under closed boundary conditions. The solid lines indicate the original dependency graph. Adding extra edges (dashed lines), we obtain a chordal graph, with maximal cliques of the form $\{1,j,j+1\}$.}
  \label{spinChain}
\end{figure}

Let us consider next the very natural experimental scenario with particles $\A=\{1,...,n\}$ distributed in a line. We are given access to the reduced density matrices of the nearest neighbors $\rho_{j j+1}$ for $j = 1,\ldots,n-1$. Namely, ${\cal I}=\{I_j\}_{j=1}^{n-1}$ with $I_{j}=\{j, j+1\}$. The dependency graph of $\Ind$, a linear graph, is again chordal, and $\{I_j\}_j$ are its maximal cliques. Therefore, the SDP hierarchies $\hier$, $\bar{\hier}$ are complete. In this case, the complexity of implementing one or the other is similar, so from now on, we will stick to $\hier$. For the sake of clarity, let us state explicitly the SDP corresponding to the ${\lev}$-th level of such hierarchy, namely 
\begin{eqnarray}\label{eq:SDPline}
 & \exists \, \,  &\lbrace \rho_{a_1\ldots a_{\lev} b_1 \ldots b_{\lev}|j} \in B(\H_j) \rbrace_{j = 1}^{n-1} \, , \nonumber \\
 & \text{s.t.} &   \forall \, \, j \, \in \, \lbrace 1,\ldots, n-1 \rbrace  : \nonumber \\
 &\text{(i)}& \rho_{a_1\ldots a_{\lev} b_1 \ldots b_{\lev}|j} \, \succeq \, 0   \nonumber \\
 & \text{(ii)} & \tr_{(a_2\ldots a_{\lev} b_2\ldots b_{\lev})} (\rho_{a_1\ldots a_{\lev} b_1 \ldots b_{\lev}|j}) = \rho_{j j+1}   \\
& \text{(iii)} & \rho_{a_1\ldots a_{\lev} b_1 \ldots b_{\lev} |j}^{T_{a_1\ldots a_k} T_{b_1 \ldots b_l}} \succeq 0 \quad k,l = 1,\ldots,{\lev}  \nonumber \\
& \text{(iv)} & \tr_{(a_1\ldots a_{\lev})} (\rho_{a_1\ldots a_{\lev} b_1 \ldots b_{\lev}|j}) \nonumber\\
&&= \tr_{(b_1\ldots b_{\lev})}(\rho_{a_1\ldots a_{\lev} b_1 \ldots b_{\lev}|j+1}) \, , \nonumber
\end{eqnarray}
\noindent where $\H_j\equiv \H_\text{sym}(d_j,{\lev}) \otimes \H_\text{sym}(d_{j+1},{\lev} ) $.

As in the star configuration, the cost of implementing the ${\lev}^{th}$ test of $\hier$ on the line, and hence certifying that the ensemble is $O(1/{\lev}^2)$-close to that of a separable state, is just polynomial on the system size $n$. In fact, even though the number of SDP matrix variables is linear in $n$, their size is \emph{independent of $n$}. This allows one to tackle systems of hundreds of sites in a normal desktop, as we see next.

For our example, we take a 1-D Ising-like model, represented by the $2$-local Hamiltonian $H=-(1/2)\sum_{j=0}^{n-1} X_jY_{j+1}$ with periodic boundary conditions. As shown in \cite{oneD}, in the limit $n\to\infty$, any quantum state with support in the (very) degenerate ground state space of this Hamiltonian is entangled.

The nearest-neighbour reduced density matrices of the thermal state $\rho_{\beta} = e^{-\beta H}/Z$ can be computed efficiently by performing a Jordan-Wigner transform and following the methods of \cite{barouch1971statistical} (see also \cite{tura2017energy}).
We wish to determine the critical inverse temperature $\beta_S$ for which the nearest-neighbour two-body marginals of $\rho_{\beta_S}$ admit a separable extension. 

Note that there exist quantum systems with more than one critical separability temperature, i.e., the thermal state of a system can be entangled for very low temperatures, separable for moderately low temperatures, entangled for moderately high temperatures and separable for very high temperatures \footnote{A two-qubit Hamiltonian, generated by random sampling, that exhibits such an aberrant thermodynamical behavior is given by $h=\left(\begin{array}{cccc}
1.3398&0.9526&0.2617&1.8461\\
0.9526&0.8519&0.2829&-1.1422\\
0.2617&0.2829&-2.2228&0.7696\\
1.8461&-1.1422&0.7696&0.0476\end{array}\right)$.}. According to our numerical tools, though, the Hamiltonian $H$, as well as the Hamiltonians that we will next introduce in this section, seems to have a single critical temperature. 

As shown in Figure \ref{fig:XY}, implementing the hierarchy \eqref{eq:SDPline} at level ${\lev} = 2$ allows us to upper bound the exact value of $\beta_S$ for systems of up to few hundred particles. More specifically, the nearest-neighbor density matrices of the thermal state of $H$ did not satisfy conditions \eqref{eq:SDPline} for inverse temperatures $\beta$ slightly greater than the values plotted in Figure \ref{fig:XY}. Curiously, increasing the level to ${\lev} = 3$ does not decrease our upper bounds on the critical temperature.

Determining (or at least upper-bounding) the critical temperature below which the thermal state of a many-body system is entangled had already been considered in several works \cite{spinModels,spinSq1,dowling2004energy,
brukner2004macroscopic,hide2007witnessing,
nakata2009thermal,hauke2016measuring,bose2005thermal}. Most of the above methods rely on a specific entanglement witness based on linear or nonlinear functions of some reduced density matrices of the system.  Our technique encompasses all these methods,  at least in the asymptotic limit of convergence of the hierarchy,  when applied to the same marginal scenario.  Let us also mention that most of the previous works focus on translationally invariant models for which we dedicate a specific comparison in the next Section.
A more interesting comparison can be made with the approach in \cite{spinSq1}, where the authors tackle this problem via spin-squeezing inequalities. To do so, they must assume that all two-body correlators are available, i.e., the authors are working in the marginal scenario $\Ind'=\{\{j,k\}:j>k\}$. Even though spin-squeezing inequalities do not fully characterize the set of all separable two-party correlators, it is interesting to see whether they are advantageous with respect to (complete) methods of entanglement detection via nearest-neighbor correlations.

In this regard, we find that spin-squeezing inequalities are not able to detect the entanglement of the thermal state of the Hamiltonian above, not even at zero temperature. On the other hand, when applied to detect the entanglement of the thermal state of the Heisenberg model, with $H=\sum_{i=1}^{n-1} (X_iX_{i+1}+Y_iY_{i+1}+Z_iZ_{i+1})$, spin-squeezing inequalities do offer an advantage, as they provide smaller upper bounds on $\beta_S$ than level $L=2$ of hierarchy (\ref{eq:SDPline}) for $n=3,...,8$. Since higher orders of the hierarchy did not seem to improve the bound on the critical temperature, we conjecture that, for some range of temperatures where the thermal state of the Heisenberg Hamiltonian is entangled, there exist separable states with the same nearest-neighbor correlators (but different two-body correlators).

The advantages of spin-squeezing inequalities seem to disappear, however, when the system under consideration is not translation-invariant. To prove this point, we consider spin-glass systems with a Hamiltonian of the form $H_{\bar{\alpha}}=\sum_i \alpha_i(X_iX_{i+1}+Y_iY_{i+1}+Z_iZ_{i+1})$, where $(\alpha_i)_i$ are random variables, independently sampled from a uniform distribution in the interval $[0,1]$. We considered systems of size $n=3,...,8$. For each system size we generated $600$ such random Hamiltonians and used spin-squeezing inequalities, as well as the level $L=2$ of hierarchy (\ref{eq:SDPline}), to ascertain which method produced smaller upper bounds on $\beta_S$. The results are shown in Table \ref{spin-glass}. As the reader can appreciate, the second level of the hierarchy (\ref{eq:SDPline}) provides better bounds for $\beta_S$ in most of the cases, and the advantage seems to grow with the system size.

\begin{table}
\begin{center}
\begin{tabular}{|c|c|}
\hline
System size&$\eta$\\
\hline
3&0.74\\
4&0.75\\
5&0.78\\
6&0.81\\
7&0.81\\
8&0.83\\
\hline
\end{tabular}
\end{center}
\caption{\label{spin-glass}Fraction of times $\eta$ that the second level of the hierarchy (\ref{eq:SDPline}) produced an upper bound on $\B_S$ (for the random Hamiltonian $H_{\bar{\alpha}}$) greater than that obtained through spin-squeezing inequalities.}
\end{table}

\subsection{Particles in a ring}

\begin{figure}
  \centering
  \includegraphics[width=8cm]{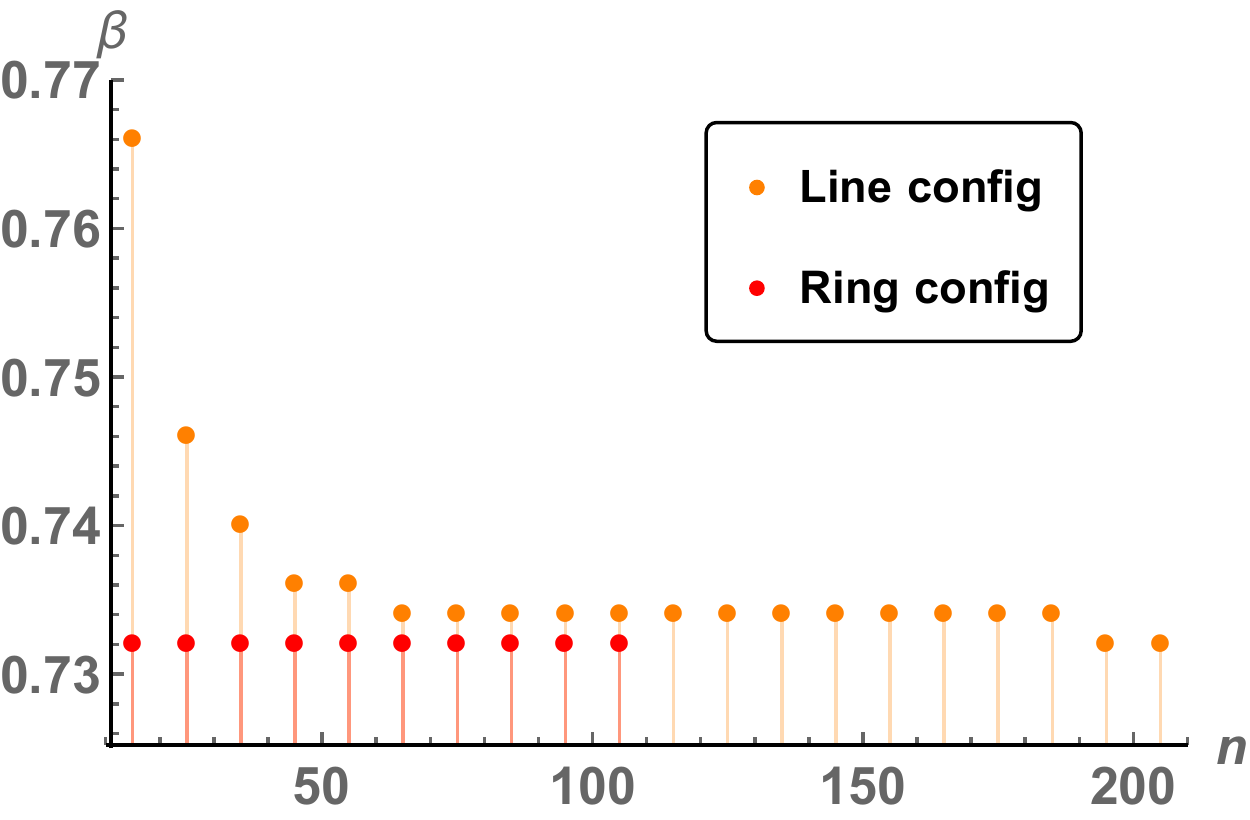}
  \caption{Upper bounds on the critical inverse temperature of separability for the thermal state of $H=-(1/2)\sum_{j=0}^{n-1} X_jY_{j+1}$, as a function of the number of particles. The results are obtained by the hierarchy $\mathbb{H}$ at the level ${\lev} = 2$ for a scenario where all the near-neighbours two-body reduced states are given. The plot compares the critical $\beta$ between the case where the information of the boundary term $\rho_{1n}$ is given (red) or not (orange). Although the hierarchy for both cases has a favourable scaling with $n$, we decided to show only results that require less than a few hours of computation in a normal laptop.}
  \label{fig:XY}
\end{figure}

Let us now complicate our 1D scenario by closing the spin chain, i.e., by assuming that we are given the knowledge also of $\rho_{1n}$. In the power set notation, such a scenario corresponds to choosing the sets of indices $I_{j}=\{j, j \mbox{ mod }n+ 1 \}$, for $j=1,...,n$. This example is useful to illustrate how one can deal with situations in which the dependency graph ${\cal G}$ of $\Ind=\{I_j\}_j$ is not chordal, see Figure \ref{spinChain}. Hence, in this scenario local compatibility of marginal probability distributions does not imply global compatibility. Thus, a modification of \eqref{eq:SDPline} by simply introducing an additional symmetric extension $\rho_{a_1\ldots a_{\lev} b_1 \ldots b_{\lev}|n}$ of the reduced density matrix $\rho_{1n}$ would lead to a sound but not convergent hierarchy of relaxations.

As explained above, one can derive a convergent SDP hierarchy for this marginal scenario by completing the original dependency graph ${\cal G}$ of $\Ind$. In this regard, Figure \ref{spinChain} shows a possible chordal completion $\bar{{\cal G}}$ for ${\cal G}$, consisting in selecting one of the nodes (say, node $1$) and adding edges with respect to any other node. The maximal cliques of the new (chordal) graph are now $\bar{I}_j \equiv \{1,j+1,j + 2\}$, for $j=1,...,n-2$. Call $\bar{{\cal I}}$ the set of all such sets. Based on these facts, to solve the 1D entanglement marginal problem for nearest-neighbour states with closed boundary conditions, we consider a slightly modified SDP hierarchy whose ${\lev}^{th}$ level is given by:
\begin{align}\label{eq:SDPring}
&\exists\{\bar{\rho}_{\bar{I}}\}_{\bar{I}\in \bar{\Ind}} \in \hier^{\lev}(\bar{\Ind})%,\rest), 
\mbox{ s.t.}\nonumber\\
&\tr_{1}\bar{\rho}_{\bar{I}_j}=\rho_{j+1j+2}, \mbox{ for } j =1,...,n-2,\nonumber\\
&\tr_{3}\bar{\rho}_{\bar{I}_1}=\rho_{12}, \tr_{n-1}\bar{\rho}_{\bar{I}_{n-2}}=\rho_{1n}.
\end{align}
This hierarchy is now complete. Moreover, by Theorem \ref{charac_chordal} and the monotonicity of the trace distance, it follows that, for any ensemble $\{\rho_I\}_{I\in\Ind}$ of states passing the $L^{th}$ level, there exists a separable state $\sigma$ with $\|\rho_I-\sigma_I\|_1\leq O\left(\frac{3d^2}{L^2}\right)$, for all $I\in\Ind$.

Remarkably, in the 1D case considered here, the polynomial scaling with $n$ of the computational effort of \eqref{eq:SDPring} is preserved, since the size of the  sets $\bar{I}_j$ is independent of $n$.
Thanks to this favourable scaling, we can use \eqref{eq:SDPring} to upper bound the critical temperatures for entanglement in the 1D model defined above for the case of periodic boundary conditions as well, the results appearing in Figure \ref{fig:XY}.

Before moving on, some comments on related work are in order. In \cite{meanField} it was noticed that, for 1D Hamiltonians with closed or open boundary conditions, the minimum energy achievable via mean field theory (namely, the separable energy) can be computed up to a fixed error with time complexity polynomial on the system size using $\epsilon$-nets. The hierarchies introduced here provide a similar result: by optimizing the energy over the set of marginals compatible with $\hier$, one obtains lower bounds to the mean field energy, whose distance to the exact results can be controlled by adjusting the hierarchy level.

\section{Infinite systems with symmetries}
\label{infiniteSec}

The previous tools can be even applied to solve the entanglement marginal problem for infinite systems, that is, with $|\A|=\infty$. Of course, the most general instance of the problem is out of reach because it requires specifying an infinite number of reduced states. However, the problem can be solved in some instances where the global system is supposed to satisfy a symmetry. The connection to the classical marginal problem plays again an important role.

\subsection{Infinite symmetric scenarios where the entanglement marginal problem is trivial}
In section \ref{oneDSec} we will show how, with minor modifications, hierarchy $\hier$ can be used to fully characterize the marginals of 1D separable TI quantum states. Before getting there, though, it is worth exploring in which situations implementing $\hier$ would be an overkill. In fact, as we next prove, if the system under study satisfies certain symmetries present in most studied models in statistical physics, then the entanglement marginal problem can be solved with relative ease.

Consider a scenario where physical systems lie on the sites of a two-dimensional square lattice. That is, $\A=\Z^2$ and $d_\alpha=d$ for all $\alpha\in\A$. We suppose the overall quantum state $\rho$ describing the systems in the lattice to be TI and also symmetric under reflections of the vertical and horizontal axes. 
This scenario appears quite frequently in statistical physics. At any temperature, the thermal state of any lattice Hamiltonian that is both translation and reflection invariant is described by a lattice state $\rho$ with these two symmetries.

In this context, consider the entanglement marginal problem with $\Ind=\{(x,y) + \hat{I}:x,y\in\Z\}$; where $\hat{I}=\{(x,y): x,y=1,2\}$, i.e., we consider all the $2\times 2$ plaquettes. By translation invariance, we have that $\rho_I=\rho_J=:\hat{\rho}$, for all $I,J\in\Ind$. The input of the entanglement marginal problem is therefore finite: it suffices to provide the $4$-partite experimental matrix $\hat{\rho}$, which we take to act on the sites in $\hat{I}$. Furthermore, by reflection symmetry we have that $\hat{\rho}$ is invariant under the permutation operators $\P_{H}$, $\P_{V}$, where $H$ ($V$) denotes the permutation that exchanges systems $(x,1)$ and $(x,2)$ for $x=1,2$ ($(1,y)$ and $(2,y)$ for $y=1,2$).

In this scenario, the marginal entanglement problem is trivial, in the following sense: $\hat{\rho}$ is the marginal of a separable TI reflection invariant (RI) quantum state iff $\hat{\rho}$ is separable and invariant under the permutations $\P_H,\P_V$. The ``only if'' implication is trivial; in the next lines we prove the ``if'' part. 

Let $\hat{\rho}$ be a $4$-partite fully separable quantum state, invariant under $\P_{H}$, $\P_{V}$. Then, there exists a measure $p_{\hat{I}}(\phi_{\hat{I}})d\phi_{\hat{I}}$ generating $\hat{\rho}$. Consider now the measure defined by 

\be
\hat{p}_{\hat{I}}(\phi_{\hat{I}})=\frac{1}{4}\sum_{\pi\in\{\id, H, V, HV\}}p_{\hat{I}}(\phi_{\pi(\hat{I})}).
\ee
\noindent This measure is symmetric under the action of $H,V$ on the $4$-dimensional vector $\phi_{\hat{I}}$. Moreover, it generates the separable state 

\begin{align}
&\frac{1}{4}\left(\hat{\rho}+\P_H\hat{\rho}\P_H^\dagger+\P_V\hat{\rho}\P_V^\dagger+\right.\nonumber\\
&\left.\P_V\P_H\hat{\rho}(\P_V\P_H)^\dagger\right)=\hat{\rho}.
\end{align}

In \cite{twoD}, it is shown that any reflection-invariant distribution $\hat{p}_{\hat{I}}(\phi_{\hat{I}})$ defined on a $2\times 2$ plaquette admits an extension $\hat{p}(\phi)$ to the whole square lattice that is also TI and RI. It follows that $\hat{\rho}$ admits a TI, RI separable extension given by $\hat{p}(\phi)$.

Unfortunately, this trick cannot be extended to characterize entanglement in, e.g., $3\times 3$ site plaquettes $I=\{(x,y):x,y=1,2,3\}$. The reason is that there exist LTI distributions $p_I(a_I)$ for the $3\times 3$ plaquette, invariant under reflections and even 90 degree rotations, which nonetheless do not admit a TI extension \cite{TI}.

The result in \cite{twoD} on the classical marginal problem generalizes to square lattices of all spatial dimensions. We arrive at the following result.

\begin{prop}
\label{trivial}
Let $\A=\Z^D$ and $\hat{I}=\{1,2\}^D$. A state $\hat{\rho}$ defined in $\hat{I}$ is the marginal of a translation and reflection invariant separable state for the whole hypercubic lattice iff $\hat{\rho}$ is fully separable and symmetric under the reflection of each orthogonal axis.
\end{prop}

In order to characterize such marginals, it is thus enough to impose symmetry and use any separability criterion, e.g., the DPS hierarchy \cite{DPS3}, to enforce separability. 

For the infinite spin chain ($D=1$), the proposition implies that the set of nearest-neighbor marginals of separable TI and RI states equals the set of separable bipartite states $\hat{\rho}_{1,2}$, with $\hat{\rho}_{2,1}=\hat{\rho}_{1,2}$. For $d=2$, the set of such states can be characterized analytically: it corresponds to the set of bipartite PPT states with symmetry under reflections. Let us remark that the spin-$1/2$ chain with translation and reflection invariance comprises the most studied integrable models in 1D, such as the XY and the Heisenberg models.

On this subject, it is worth noting that in \cite{spinModels} T\'oth shows how to compute the exact separable energy per site of 2-local translation-invariant Hamiltonians with reflection symmetry. For $D=1$ systems, he observes that the separable energy can be achieved by states of the form $...\ket{\psi_1}\ket{\psi_2}\ket{\psi_1}\ket{\psi_2}...$. This can be regarded as the dual of our result, since the extreme points of symmetric separable bipartite states are of the form $\hat{\rho}_{12}=\frac{1}{2}(\proj{\psi_1}\otimes\proj{\psi_2}+\proj{\psi_2}\otimes\proj{\psi_1})$, and thus $\tr(\hat{\rho}_{12}H_{12})=\tr\{(\proj{\psi_1}\otimes\proj{\psi_2})H\}$, for all operators $H$ such that $H_{12}=H_{21}$. For higher spatial dimensions, Proposition \ref{trivial} extends T\'oth's results from Hamiltonians with nearest-neighbor interactions to Hamiltonians with general (reflection-symmetric) plaquette operators.

\subsection{Translation invariance in one dimension}
\label{oneDSec}

What happens when we care about next-to-nearest neighbor correlations, or lose invariance under reflections? Then the entanglement marginal problem becomes non-trivial, and SDP hierarchies such as $\hier$ are needed. In this section, we study how to solve the entanglement marginal problem for infinite TI 1D chains.

Consider thus a scenario where quantum systems of dimension $d$ rest on the sites of an infinite spin chain, i.e., $\A=\Z$. As in the previous section, we assume that the overall quantum state is TI, but not necessarily invariant under reflections. The input of the entanglement marginal problem is $\hat{\rho}$, the reduced density matrix of $k$ consecutive neighbors, i.e., $\hat{\rho}=\rho_{\hat{I}}$, with $\hat{I}=\{1,...,k\}$. We wish to determine whether $\hat{\rho}$ is the marginal of a separable, TI quantum state.

Suppose, for a moment, that such is the case. In \cite{oneD} it is shown that, should there exist an overall TI separable state, the global measure $p(\phi)d\phi$ in eq. (\ref{separable}) can also be chosen TI. Hence the $k$-variate distribution $p_{\hat{I}}$ generating the separable state $\hat{\rho}$ can be assumed LTI. This observation implies that $\hat{\rho}^{({\lev})}$, as defined in (\ref{truqui}) is also LTI, i.e., it satisfies

\be
\tr_{\{1\}^{\lev}}(\hat{\rho}^{({\lev})})=\tr_{\{k\}^{\lev}}(\hat{\rho}^{({\lev})}).
\label{LTI_quantum}
\ee

This suggests tackling the entanglement marginal problem with an SDP hierarchy $\hat{\hier}$ whose ${\lev}^{th}$ level is given by the SDP:
\begin{align}
&\exists \hat{\rho}^{({\lev})}\in B(\H_\text{sym}^{{\lev},d})^{\otimes k}, \hat{\rho}^{({\lev})} \, \succeq 0 \, \, \, \text{ s.t.}\nonumber\\
&\tr_{\hat{I}^{{\lev}-1}}(\hat{\rho}^{({\lev})})=\hat{\rho},\nonumber\\
&\tr_{\{1\}^{\lev}}(\hat{\rho}^{({\lev})})=\tr_{\{k\}^{\lev}}(\hat{\rho}^{({\lev})}).
\label{TI_SDP}
\end{align}

\begin{theo}
Let $\hat{\rho}\in \hat{\hier}^\lev$. Then, there exists a translation invariant separable state $\sigma$ for the infinite spin chain such that

\be
\|\hat{\rho}-\sigma_{\hat{I}}\|_1\leq O\left(\frac{kd^2}{\lev^2}\right).
\ee
\end{theo}

\begin{proof}
In the proof of Proposition \ref{conv_prop} (see Appendix \ref{proofProp}), we applied the linear maps defined in \cite{NOP} to the matrix variables $\{\rho^{({\lev})}_I\}_I$ appearing in the definition of $\hier$ in order to generate the separable ensemble $\{\tilde{\rho}_I\}_I$. As explained in Appendix \ref{proofProp}, these linear maps have the convenient property of translating linear relations over the matrix variables $\{\rho^{({\lev})}_I\}_I$ into analogous linear relations over the distributions $\{\tilde{p}_I(\phi_I)\}_I$ generating $\{\tilde{\rho}_I\}_I$. Consequently, if we apply said maps to $\hat{\rho}^{({\lev})}$, we find that, for any $\hat{\rho}\in\hat{\hier}^{\lev}$, there exists a separable state $\tilde{\rho}$, with $\|\tilde{\rho}-\hat{\rho}\|_1\leq O\left(\frac{kd^2}{{\lev}^2}\right)$ generated by a measure $\tilde{p}_{\hat{I}}(\phi_{\hat{I}})$ such that

\be
\int d\varphi\tilde{p}_{\hat{I}}(\varphi,\phi_2,...,\phi_{k})=\int d\varphi\tilde{p}_{\hat{I}}(\phi_1,...,\phi_{k-1}, \varphi).
\label{LTI}
\ee
\noindent That is, the LTI of the variable $\hat{\rho}^{({\lev})}$, enforced by the last relation in (\ref{TI_SDP}), is inherited by the distribution $\tilde{p}_{\hat{I}}$ generating the separable approximation to $\hat{\rho}$.

It is a well-known result in statistical physics (see, e.g., \cite{oneD} for a proof) that every $k$-variate distribution $\tilde{p}_{\hat{I}}$ satisfying 1D LTI admits a 1D TI extension $\tilde{p}$. Hence, $\tilde{\rho}$ admits a TI separable extension $\sigma$, and we arrive at the statement of the theorem.

\end{proof}

To test how the SDP hierarchy $\hat{\hier}$ performs in practice, we take the entanglement witness for translation-invariant states presented in \cite{oneD}. There, it is shown that, if $\rho_{1,2}$ is the nearest-neighbor reduced density matrix of a separable infinite TI state, then it must satisfy $\tr\{\rho_{12}(X\otimes Z)\}\geq -\frac{1}{2}$. Moreover, this bound is tight. We verify that a second order relaxation ${\lev}=2$ recovers this bound, up to the solver's precision. Note that, since $X\otimes Z$ is not invariant under reflection, we cannot invoke Proposition \ref{trivial} to solve this problem. Moreover, as observed in \cite{oneD}, $\tr\{\sigma_{12}(X\otimes Z)\}=-1$ for the separable and LTI state $\sigma_{12}=\frac{1}{2}(\proj{+}\otimes\proj{-i}+\proj{-}\otimes\proj{+i})$, with $X\ket{\pm}=\pm\ket{\pm}, Y\ket{\pm i}=\pm\ket{\pm i}$.

\subsection{2D translation invariance}

Extending the 1D argument for TI systems to the 2D case without further symmetries is non-trivial. To begin, using the symmetrization technique introduced in \cite{oneD}, one can prove that the marginal of a 2D separable TI state admits a separable decomposition generated by a TI measure. Hence, for $\hat{I}=\{(x,y): 1\leq x\leq k, 1\leq y\leq l\}$, we can assume that the distribution $p_{\hat{I}}$ in (\ref{truqui}) is LTI. Hence the corresponding SDP variable $\hat{\rho}^{({\lev})}$ is subject to the LTI conditions:

\begin{align}
&\tr_{\{(1,y)^{\lev}:y\}}(\hat{\rho}^{({\lev})})=\tr_{\{(k,y)^{\lev}:y\}}(\hat{\rho}^{({\lev})}),\nonumber\\
&\tr_{\{(x,1)^{\lev}:x\}}(\hat{\rho}^{({\lev})})=\tr_{\{(x,l)^{\lev}:x\}}(\hat{\rho}^{({\lev})}).
\label{LTI_quantum2}
\end{align}

\noindent Applying the linear maps from \cite{NOP} over any feasible matrix variable $\hat{\rho}^{({\lev})}$, we arrive at a separable decomposition for $\hat{\rho}$ generated by a LTI distribution $\tilde{p}(\phi_{\hat{I}})$.

This is the point where problems appear. Unfortunately, in two spatial dimensions LTI does \emph{not} imply the existence of a global TI measure~\cite{TI}. In fact, the simplest problem of deciding whether a pair of distributions $P_h,P_v$ are the horizontal and vertical nearest-neighbor distributions of a 2D translation-invariant model in the square lattice is undecidable \cite{undecidable}. Defining a complete hierarchy of SDPs to characterize entanglement in the 2D TI case would require considering sequences of extensions $(\hat{\rho}_{I_j}^{({\lev})})_j$ involving an ever-growing number of sites $I_j$.

A compromise between full reflection symmetry and raw TI is a scenario where the considered system is only symmetric with respect to reflections over the vertical axis, see Figure \ref{squareLattice}. Let us then consider an infinite square lattice, call $\hat{I}$ the $l\times 2$ plaquette $\{(x,y): x=1,2, y=1,...,l\}$, and let $\hat{\rho}$, defined on the sites $\hat{I}$, be the input of an entanglement marginal problem. We denote by $U, L$ the subsets of $\hat{I}$ given by $\{(x,y):x=1,2; y=1,...,k-1\}$, $\{(x,y):x=1,2; y=2,...,k\}$, respectively.

\begin{figure}
  \centering
  \includegraphics[width=6cm]{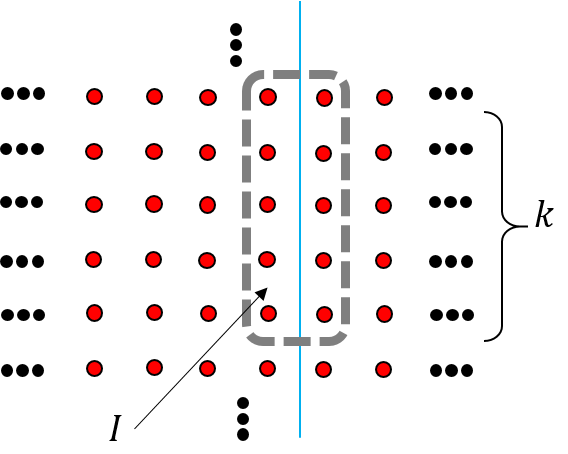}
  \caption{\textbf{The 2D case.} Subsystems (in red) are located at each site of an infinite square lattice, subject to symmetry under reflections of the blue axis and translation invariance. The set $\hat{I}$ of sites is circled with dashed lines.}
  \label{squareLattice}
\end{figure}

Acknowledging the LTI and invariance under reflections over the vertical axis of the underlying distribution $\hat{p}(\phi)$ generating the separable state $\hat{\rho}$, we arrive, via (\ref{truqui}), at a hierarchy of SDP relaxations whose ${\lev}^{th}$ level is:

\begin{align}
&\exists \hat{\rho}^{({\lev})}\in B(\H_\text{sym}^{{\lev},d})^{\otimes 2k},\hat{\rho}^{({\lev})}\succeq 0,\nonumber\\
&\tr_{\{(1,1),(2,1)\}^{\lev}}(\hat{\rho}^{({\lev})})= \tr_{\{(1,k),(2,k)\}}^{\lev}(\hat{\rho}^{({\lev})}),\nonumber\\
&\mbox{SWAP}\hat{\rho}^{({\lev})}\mbox{SWAP}^\dagger = \hat{\rho}^{({\lev})},
\end{align}
\noindent where $\mbox{SWAP}$ is the operator permuting the Hilbert spaces of sites $(0,y)$ and $(1,y)$, for $y=0,...,k-1$. The convergence of this hierarchy follows from the proof of proposition \ref{conv_prop} and the result, proven in \cite{twoD}, that any distribution $P_{\hat{I}}(a_{\hat{I}})$ satisfying reflection symmetry along the vertical axis in addition to LTI admits a TI extension that is also symmetric with respect to vertical reflections. 

We arrive at the following theorem.

\begin{theo}
Let $\hat{\rho}$ satisfy the $\lev^{th}$ level of the above hierarchy. Then, there exists a separable state $\sigma$ for the infinite square lattice, invariant under translations and vertical reflections, such that

\be
\|\sigma_{\hat{I}}-\hat{\rho}\|\leq O\left(\frac{kd^2}{\lev^2}\right).
\ee
\end{theo}

The characterization of TI distributions with reflection symmetry presented in \cite{twoD} can be easily extended to arbitrary spatial dimensions. The general result would state that any probability distribution $P$ for the variables of the sites $\{(x_1,x_1,...,x_{D}):x_1\in \{1,...,k\},x_2,...,x_{D}\in\{1,2\}\}$ for a hypercubic lattice admits a TI extension with symmetry under the reflection of the last $D-1$ axes iff $P$ is LTI along the first axis and symmetric under the inversion of the rest. Implementing the $2^{nd}$ level of the corresponding SDP hierarchy to tackle the 3D entanglement marginal problem is, however, unrealistic in a normal desktop.

%\vspace{10pt}
%\noindent\emph{Conclusion}
\section{Conclusion}
\label{conclusion}
In this work, we have extensively studied the entanglement marginal problem, namely the problem of detecting the entanglement of an overall quantum state given access to some ensemble of reduced density matrices. We have presented a general framework to derive hierarchies of SDP relaxations of the set of ensembles admitting a separable extension. We related the convergence or completeness of these hierarchies to the classical marginal problem, and exploited this connection to solve the entanglement marginal problem for condensed matter systems of arbitrary size in 1D and also in 2D (under reflection symmetries). In addition, the introduced tools allowed us to show that, for some thermodynamical scenarios, the entanglement marginal problem becomes trivial. In fact, the set of two-qubit nearest-neighbor marginals of a 1D TI, RI separable state can be characterized analytically.

Remarkably, the computational complexity of the methods presented in this work scales very well with the system size, when it depends on it at all. We are thus optimistic that our results will find application to assess in a scalable way the entanglement of the states generated in current and near-term quantum setups, such as Noisy Intermediate Scale Quantum (NISQ) devices.

\vspace{10pt}
\noindent \emph{Acknowledgements}

MN was supported by the Austrian Science fund (FWF) stand-alone project P 30947. FB was supported by the Deutsche Forschungsgemeinschaft (DFG, German Research Foundation) – Project number 414325145 in the framework of the Austrian Science Fund (FWF): SFB F71. AA was supported by the ERC AdG CERQUTE, the AXA Chair in Quantum Information Science, the Spanish MINECO
(Severo Ochoa grant SEV-2015-0522), the
Generalitat de Catalunya (CERCA Program, QuantumCAT and SGR 1381) and the Fundaci\'o{}
Privada Cellex and Mir-Puig.

\bibliographystyle{apsrev4-1}
\bibliography{ent1DBib_doi}

\begin{appendix}

\section{Infinite systems}
\label{infinite_systems_app}
In this section, we provide simple operational definitions for classical distributions, quantum states and separable quantum states describing physical systems with a countable set of subsystems $\A$, finite or infinite.

\begin{defin}
\noindent\textbf{Quantum states}

Let $\A$ be a countable alphabet, and, for each $\alpha\in\A$, let $d_\alpha\in\mathbb{N}$ denote the Hilbert space dimension of system $\alpha$. Let $\rho$ be a function mapping every finite subset ${\cal B}\subset \A$ into a trace-one positive semidefinite matrix $\rho_{{\cal B}}\in \bigotimes_{\beta\in{\cal B}} B(\C^{d_{\beta}})$, with the property

\be
\tr_{{\cal B}\setminus {\cal C}}(\rho_{{\cal B}})=\tr_{{\cal C}\setminus {\cal B}}(\rho_{{\cal C}}),
\ee
\noindent for all finite ${\cal B},{\cal C}\subset \A$. 

We say that $\rho$ is a \emph{quantum state} for system $\A$.
\end{defin}

For finite $\A$, this definition coincides with the standard definition of a quantum state. For infinite alphabets, though, we regard a quantum state as an infinite set of reduced density matrices, each acting on finitely many systems. One can construct an analogous definition for classical probability distributions, by replacing positive semidefinite matrices with classical measures and partial traces with integrals over finite sets of random variables.

The definition of separability in infinite dimensional systems follows similar lines.

\begin{defin}
\noindent\textbf{Separable quantum states}

Let $\A$ be a countable alphabet, and let $\rho$ be a state of $\A$. If, for all finite subsets ${\cal B}\subset \A$, $\rho_{{\cal B}}$ is a separable state, then we say that $\rho$ is separable.
\end{defin}

Using a Cantor-like argument, one can see that $\rho$ is a separable state iff there exists a distribution $p$ such that, for all finite ${\cal B}\subset \A$,

\be
\rho_{{\cal B}}=\int d\phi_{{\cal B}}p_{{\cal B}}(\phi_{{\cal B}})\bigotimes_{\beta\in{\cal B}}\proj{\phi_\beta}.
\label{global_separable}
\ee

This is evident if $\A$ is finite. To construct such a distribution in the $|\A|=\infty$ case (more specifically: a function $p$ mapping any finite subset ${\cal B}\subset\A$ into a distribution $p_{{\cal B}}(\phi_{{\cal B}})$ such that eq. (\ref{global_separable}) holds and such that, for any finite $\B,{\cal C}\subset \A$, $p_{\cal B},p_{\cal C}$ are locally compatible), we start from the premise that all states $\rho_{\B}$ are separable. Consequently, for all finite $\B\subset \A$, there exist (not necessarily locally compatible) distributions $\tilde{p}_{\cal B}$ satisfying 

\be
\rho_{{\cal B}}=\int d\phi_{{\cal B}}\tilde{p}_{{\cal B}}(\phi_{{\cal B}})\bigotimes_{\beta\in{\cal B}}\proj{\phi_\beta}.
\label{global_separable_tilde}
\ee
\noindent Next we consider an increasing sequence $(\B^i)_i$ of finite subsets of $\A$ such that $\cup_i\B^i=\A$, as well as the sequence of distributions $\tilde{p}_i=\tilde{p}_{\B_i}$, and partition the measure space of the variables $\phi_{\B_i}$ into a finite set of measurable subsets ${\cal F}_i$, with ${\cal F}_i={\cal F}_{i-1}\times {\cal F}'_i$. By the Bolzano-Weierstrass theorem, there exists a subsequence $(s^1_j)_j$ of $(i)_i$ such that the limit

\be
p_{{\cal B}_1}(\Phi):=\lim_{j\to\infty}\int_{\Phi} d_{\phi_{\B_1}}\int d\phi_{\B_{s^1_j}\setminus \B_1}\tilde{p}_{s^1_j}(\phi_{\B_{s^1_j}})
\ee
\noindent exists, for all $\Phi\in {\cal F}_1$. Now, starting from the sequence $(s^1_j)_j$, we invoke the Bolzano-Weierstrass theorem again to show that there exists a subsequence $(s^2_j)_j$ of $(s^1_j)_j$ with the property that 

\begin{align}
&p_{{\cal B}_2}(\Phi)=\lim_{j\to\infty}\int_{\Phi_1\times\Phi'_2} d_{\phi_{\B_2}}\int d\phi_{\B_{s^2_j}\setminus \B_2}\tilde{p}_{s^2_j}(\phi_{\B_{s^2_j}})\nonumber\\
&=:p_{\B_2}(\Phi_1\times \Phi_2')
\end{align}
\noindent exists for all $\Phi_1\times \Phi_2'\in {\cal F}_2$. Moreover, since $(s^2_j)_j$ is a subsequence of $(s^1_j)_j$, we have that 

\be
\sum_{\Phi_2'}p_{\B_2}(\Phi_1\times \Phi_2')=p_{\B_1}(\Phi_1).
\ee
\noindent That is, within the coarse-graining given by ${\cal F}^2$, the distributions $p_{\B_1}$ and $p_{\B_2}$ so defined are locally compatible.

Iterating this argument, we arrive at a sequence of nested subsequences with the property that, for the $k^{th}$ subsequence $(s^k_j)_j$, the limit

\be
\lim_{j\to \infty}\int_{\Phi} d\phi_{\B_k}\int_{\B_{s^k_j}\setminus\B_k}\tilde{p}_{\B_{s^k_j}}(\phi_{\B_{s^k_j}})=:p_{\B_k}(\Phi)
\ee
\noindent exists for all $\Phi\in {\cal F}_k$, and all so-generated coarse-grained probability distributions $\{p_{\B_k}(\Phi):\Phi\in {\cal F}_k$ are locally compatible. Moreover, if we consider the sequence $(s_j)_j\equiv(s^j_j)_j$, we have that, for all $k$,

\be
\lim_{j\to \infty}\int_{\Phi} d\phi_{\B_k}\int_{\B_{s_j}\setminus\B_k}\tilde{p}_{\B_{s_j}}(\phi_{\B_{s_j}})=p_{\B_k}(\Phi).
\ee

Within the coarse-grainings ${\cal F}^k$, we hence arrive at a global definition of $p$. To arrive at a more fine-grained description of $p$, we consider sub-partitions $\bar{{\cal F}}_k$ of ${\cal F}_k$ such that $\bar{{\cal F}}_k=\hat{{\cal F}}_{k-1}\times \hat{{\cal F}}_{k}'$ and apply the same argument as before over the sequence $(s_j)_j$. Iterating the argument over and over, we find that there exists a sequence $(\hat{s}_j)_j$ such that 

\be
\lim_{j\to \infty}\int_{\Phi} d\phi_{\B_k}\int_{\B_{\hat{s}_j}\setminus\B_k}\tilde{p}_{\B_{\hat{s}_j}}(\phi_{\B_{\hat{s}_j}})=:p_{\B_k}(\Phi)
\ee
\noindent exists for any measurable set $\Phi$. $(p_{\B_j})_j$ are, by construction, locally compatible. Moreover, since $(\tilde{p}_{\B_j})_j$ satisfy eq. (\ref{global_separable_tilde}), it follows that $p$ satisfies (\ref{global_separable}).

\section{A simplified hierarchy for the entanglement marginal problem}
\label{simplified}
Suppose that there exists a subset $\Ind_+$ of $\Ind$ such that any $I\in \Ind_+$ contains a letter $\alpha_I$ that is not present in any other set $J\in\Ind$. Define $\bar{I}$ as $I\setminus\{\alpha_I\}$ for $I\in\Ind_+$ or $I$, otherwise, and, for any $I\not\in\Ind_+$, let $\alpha_I\not\in\A$ denote a 1-dimensional Hilbert space. This notation will become clear below.

Let $\{\rho_I\}_I$ admit a separable extension generated by the measure $p(\phi)$, and define the states

\be
\bar{\rho}^{({\lev})}_{I}\equiv\int p_{I}(\phi_I)d\phi_I\bigotimes_{\alpha\in \bar{I}}\proj{\phi_\alpha}^{\otimes {\lev}}\otimes \proj{\phi_{\alpha_I}}.
\label{truqui2}
\ee

\noindent Note that whenever $I\not\in\Ind_+$, the last element $\alpha_I$ in the previous equation corresponds to a one-dimensional Hilbert space and can in practice be removed. It is however convenient to write it to have a unified notation for the two cases $I\in\Ind_+$ and $I\not\in\Ind_+$. Then the states $\{\bar{\rho}^{({\lev})}_I\}$ satisfy the conditions

\begin{enumerate}[label=(\roman*)]
\item
$\tr_{\bar{I}^{{\lev}-1}}(\bar{\rho}^{({\lev})}_{I})=\rho_{I}$.

\item
$\bar{\rho}^{({\lev})}_{I}$ is Positive under Partial Transposition (PPT) across all bipartitions.
\item
$\bar{\rho}^{({\lev})}_{I}\in B(\bigotimes_{\alpha\in \bar{I}}\H_\text{sym}({\lev},d_\alpha)\otimes\H_{\alpha_I})$.
\item
%$\tr_{\alpha^{\lev}: \alpha\in I\setminus J}(\rho^{({\lev})}_{I})=\tr_{\alpha^{\lev}: \alpha\in J\setminus I}(\rho^{({\lev})}_{J})$.
$\tr_{(I\setminus J)^{\lev},\alpha_I}(\bar{\rho}^{({\lev})}_{I})=\tr_{(J\setminus I)^{\lev},\alpha_J}(\bar{\rho}^{({\lev})}_{J})$ (after appropriately reordering the systems of one of the sides).
%\item
%For all $({\cal P}, I, c_{{\cal P}})\in\rest$, $\sum_{\pi\in{\cal P}}c_\pi \rho^{({\lev})}_{\pi(I)}=0$.
\end{enumerate}

As in the definition of $\hier$, we relax the property that $\{\rho_I\}_I$ admit a separable extension by demanding that there exist positive semidefinite matrices $\{\bar{\rho}^{({\lev})}_I\}_I$ satisfying conditions (i)-(iv). The resulting SDP hierarchy $\bar{\hier}(\Ind|\{\rho_I\}_I)$ is weaker than $\hier(\Ind|\{\rho_I\}_I)$, in the sense that $\{\rho_I\}_I\in \hier^{\lev}$ implies $\{\rho_I\}_I\in \bar{\hier}^{\lev}$. However, for some configurations of $\Ind$, verifying that $\{\rho_I\}_I\in\bar{\hier}^{\lev}(\Ind)$ requires considerably less computational resources. This is because, as said, the index $\alpha_I$ is not extended in this hierarchy.

\section{From matrix variables to separable states and probability distributions}
\label{proofProp}
In this Appendix, we prove Proposition \ref{conv_prop}. We also show that linear relations over the matrix variables in the hierarchy $\hier$ described in the main text are inherited by the measure generating the separable approximation to $\{\rho_I\}_I$. For simplicity, we first prove that Proposition \ref{conv_prop} holds for any ensemble $\{\rho_I\}_I\in\hier^{\lev}$. The extension to the weaker condition $\{\rho_I\}_I\in\bar{\hier}^{\lev}$, defined in Appendix \ref{simplified}, will follow easily. 

To analyze the speed of convergence of the DPS hierarchy~\cite{DPS1, DPS2, DPS3}, one of us introduced in~\cite{NOP} a family of linear maps $\Omega^{d,{\lev}}_\phi:B(\H_\text{sym}^{{\lev},d})\to B(\C^d)$ with the property that, for all Hilbert spaces $\H$ and all states $\omega_{1^{\lev}A}\in B(\H_\text{sym}({\lev},d)\otimes \H)$ positive semidefinite under the transposition of $\lfloor \frac{{\lev}}{2}\rfloor$ of the ${\lev}$ symmetric systems, the following inequality holds:

\be
(\Omega^{d,{\lev}}_{\phi}\otimes \id_{B(\H)})(\omega_{1^{\lev}A})=\proj{\phi}_1\otimes (\omega_{\phi})_{A}\geq 0.
\label{posi}
\ee
\noindent Furthermore, the map $\Omega^{d,{\lev}}\equiv \int d\phi\Omega^{d,{\lev}}_{\phi}$ satisfies

\begin{align}
(\Omega^{d,{\lev}}\otimes \id_{B(\H)})(\omega_{1^{\lev}A})=&(1-\epsilon(\lev,d))\omega_{1^{\lev}A} +\nonumber\\&\epsilon(\lev,d)\frac{\id_1}{d}\otimes\omega_{A},
\label{approx}
\end{align}
\noindent with $\epsilon(\lev,d)$ defined as in (\ref{epsilon_def}).

Now, suppose that the ensemble $\{\rho_{I}\}_{I}$ admits extensions $\{\rho_I^{({\lev})}\}$ satisfying the conditions (i)-(iv) in the main text. We simultaneously apply the maps $\int d\phi\Omega^{d_\alpha,{\lev}}_\phi$, for $\alpha\in \A$, to each of the states $\{\rho_I^{({\lev})}\}$. This results in the states $\tilde{\rho}_I=\int d\phi_I \tilde{p}(\phi_I)\bigotimes_{\alpha\in I}\proj{\phi_\alpha}$, with $\tilde{p}(\phi_I)= \tr(\bigotimes_{\alpha\in I}\Omega^{d_\alpha,{\lev}}_{\phi_i}\rho_I^{({\lev})})$. Note that, if $\{\rho^{({\lev})}_{I}\}_{I}$ satisfies linear constraints of the form 

\be
\sum_{I\in\Ind,J\subset I}c_{IJ}\tr_{(I/J)^{\lev}}\rho^{({\lev})}_{I}=0,
\ee

\noindent such as local compatibility, or relations (\ref{LTI_quantum}), (\ref{LTI_quantum2}), then so will the ensemble of measures $\{\tilde{p}_I\}_I$. Finally, due to eq. (\ref{approx}), 

\begin{align}
&\tilde{\rho}_I=\prod_{\alpha\in I}(1-\epsilon(d_\alpha,\lev))\rho_I +\nonumber\\
& \left(1-\prod_{\alpha\in I}(1-\epsilon(d_\alpha,\lev))\right)\sigma_I,
\end{align}
\noindent for some normalized state $\sigma_I$. Subtracting $\rho_I$ and invoking the triangle inequality, we arrive at equation (\ref{poly_approx}) in Proposition \ref{conv_prop}.

To see that Proposition \ref{conv_prop} also applies to the simplified hierarchy $\bar{\hier}$ with no linear constraints on $\{\rho_I^{({\lev})}\}_I$ besides local compatibility, just apply to each variable $\bar{\rho}^{({\lev})}_I$ the maps $\int d\phi\Omega^{d_\alpha,{\lev}}_\phi$, for $\alpha\in\bar{I}$. The result will be a separable state of the form 

\be
\tilde{\rho}_I=\int p_{\bar{I}}(\phi_{\bar{I}})\left(\bigotimes_{\alpha\in \bar{I}}\proj{\phi_\alpha}\right)\otimes \sigma(\phi_{\bar{I}}),
\ee
\noindent where $\sigma(\phi_{\bar{I}})$ is the state of system $\alpha_I$ after we subject systems in $\bar{I}$ to the map $\bigotimes_{\alpha\in \bar{I}}\Omega^{d_\alpha,{\lev}}_{\phi_{\alpha}}$. Writing 

\be
\sigma(\phi_{\bar{I}})=\int d\phi_{\alpha_I}\tilde{p}(\phi_{\alpha_I}|\phi_{\bar{I}})\proj{\phi_{\alpha_I}},
\ee
\noindent we find that the separable state $\tilde{\rho}_I$ is generated by the distribution $\tilde{p}(\phi_I)\equiv \tilde{p}_{\bar{I}}(\phi_{\bar{I}})\tilde{p}(\phi_{\alpha_I}|\phi_{\bar{I}})$, with 

\be
\tilde{p}_{\bar{I}}(\phi_{\bar{I}})=\tr\left(\bigotimes_{\alpha\in \bar{I}}\Omega^{d_\alpha,{\lev}}_{\phi_{\alpha}}(\bar{\rho}^{({\lev})})\right).
\ee

Since, by assumption, system $\alpha_I$ does not belong to any other $J\in\Ind, J\not=I$, local compatibility of $\{\bar{\rho}^{({\lev})}_I\}_I$ implies local compatibility of the distributions $\{\tilde{p}_{I}\}_I$.

\section{Global compatibility for chordal graphs}
\label{runningApp}

The purpose of this section is to prove that, if the dependency graph ${\cal G}$ of $\Ind$ is chordal and the elements of $\Ind$ correspond to its maximal cliques, then local compatibility of distributions $\{p_I(\phi_I)\}_{I\in\Ind}$ implies global compatibility.

It is a standard result in combinatorics \cite{chordal} that the maximal cliques of a chordal graph can be ordered as $I_1,...,I_m$ in such a way that, for all $j\in\{1,...,m-1\}$, 

\be
I_{j+1}\cap \left(\bigcup_{k=1}^{j} I_k \right)\subset I_s,
\label{running}
\ee
\noindent for some $s\leq j$. This condition is known as the \emph{running intersection property}.

We next prove that the running intersection property, together with local compatibility, implies the existence of a global distribution $p(\phi)$. Call $V, \Lambda$ the union and intersection of the sets $I_1, I_2$, respectively. Note that, given the sets $\{p_I(\phi^{I})\}_I$, $p(\phi_{\Lambda})$ is well defined by local compatibility. Now, define $p_V(\phi_{V}):=p(\phi_{I_1})p(\phi_{I_2})/p(\phi_{\Lambda})$. Summing/integrating over the variables $\phi_i$, with $i\in I_2\setminus \Lambda$, we obtain $p(\phi_{I_1})$. Similarly, summing/integrating over $\phi_i$, with $i\in I_1\setminus \Lambda$, we obtain $p_{I_2}(\phi_{I_2})$. Hence $p_V$ admits $p_{I_1}, p_{I_2}$ as marginals. Furthermore, consider any set $I_j$, with $j>2$. By the running intersection property, either $I_j\cap V\subset I_1$ or $I_j\cap V\subset I_2$, and so $p_{I_j\cap V}(\phi_{I_j\cap V})$ is well defined. It thus follows that the sets of indices $\{V\}\cup \{I_j\}_{j=2}^m$ and the distributions $\{p_V\}\cup \{p_{I_j}\}_{j=3}^m$ satisfy, respectively, the running intersection property and local compatibility. Iterating this procedure $m-2$ times, we arrive at a global probability distribution.

We next illustrate the construction of the overall distribution with an example. Think of the open spin chain discussed in Section \ref{lineSec}, where $\Ind=\{I_j:j =1,...,n-1\}$, with $I_j=\{j,j+1\}$. It can be verified that the sequence of sets $(I_j)_j$ satisfies the running intersection property (\ref{running}). Now, let $\{p_{\{j,j+1\}}(\phi_j,\phi_{j+1})\}_{j=1}^{n-1}$ be a set of locally compatible distributions. Following the procedure sketched above, a global distribution for the variables at sites $1,2,3$ is given by

\be
p_{\{1,2,3\}}(\phi_1,\phi_2,\phi_3)=\frac{p_{\{1,2\}}(\phi_1,\phi_2)p_{\{2,3\}}(\phi_2,\phi_3)}{p_2(\phi_2)},
\label{global}
\ee
\noindent where

\begin{align}
p_2(\phi_2):=&\int d\phi_1p_{\{1,2\}}(\phi_1,\phi_2)=\nonumber\\
&\int d\phi_3p_{\{2,3\}}(\phi_2,\phi_3).
\label{lc_system2}
\end{align}
\noindent Note that the last relation holds by local compatibility.

To see why (\ref{global}) is an extension of $p_{\{1,2\}}, p_{\{2,3\}}$, integrate $\phi_1$ in both sides of (\ref{global}). We find, by (\ref{lc_system2}), that  

\be
\int d\phi_1p_{\{1,2,3\}}(\phi_1,\phi_2,\phi_3)=p_{\{2,3\}}(\phi_2,\phi_3).
\ee
\noindent Similarly, integrating $\phi_3$ we arrive at

\be
\int d\phi_3p_{\{1,2,3\}}(\phi_1,\phi_2,\phi_3)=p_{\{1,2\}}(\phi_1,\phi_2).
\ee

By increasing the number of variables of the overall probability distribution in this fashion, we find that 

\be
p(\phi):= p_{\{1, 2\}}(\phi_{1},\phi_{2})\prod_{j=2}^{n-1}\frac{p_{\{j, j+1\}}(\phi_{j},\phi_{j+1})}{p_{j}(\phi_{j})}
\ee
\noindent is an extension of the ensemble $\{p_{\{j, j+1\}}:j=1,...,n-1\}$.

\end{appendix}

\end{document}